\documentclass[a4paper,10pt]{article}
\usepackage{arxivj}
\usepackage{amsmath,amssymb,graphicx,fancybox,ascmac,algorithm,algorithmic,cite,nidanfloat,amsthm}
\usepackage[FIGTOPCAP]{subfigure}

\newcommand{\bm}[1]{\mbox{\boldmath{$#1$}}} 
\newtheorem{theorem}{Theorem}



\title{Adaptive Markov Chain Monte Carlo for Auxiliary Variable Method and Its Application to Parallel Tempering}
\author{Takamitsu Araki \and
        Kazushi Ikeda }
\affiliation{
Graduate School of Information Science, Nara Institute of Science and
Technology}
\email{takamitsu-a@is.naist.jp, kazushi@is.naist.jp}

\begin{document}

\maketitle

\begin{abstract}
Auxiliary variable methods such as the Parallel Tempering and the cluster Monte Carlo methods 
generate samples that follow a target distribution by using proposal and auxiliary distributions.
In sampling from complex distributions, these algorithms  are highly more efficient than the standard Markov chain Monte Carlo methods. 
However, their performance strongly depends on their parameters and determining the parameters is critical.
In this paper, we proposed an algorithm for adapting the parameters during drawing samples 
and proved the convergence theorem of the adaptive algorithm. 
We applied our algorithm to the Parallel Tempering. That is, we developed adaptive Parallel Tempering that tunes the parameters on the fly. 
We confirmed the effectiveness of our algorithm through the validation of the adaptive Parallel Tempering, comparing samples from the target distribution by the adaptive Parallel Tempering and samples by conventional algorithms.

\end{abstract}


\section{Introduction}

Markov chain Monte Carlo (MCMC) methods have been an important algorithm in various scientific fields (Liu 2001, Robert and Casella 2004).
MCMC methods can generate samples that follow a target distribution
by using a simple proposal distribution.
However, 
in sampling from complex distribution such as multi-modal, 
the standard MCMC methods produce samples theoretically converge the target distribution but practically do not.
The produced samples can be trapped in a local mode for a extremely long period.

To cope with this localization problem, 
the parallel tempering (PT) a.k.a.\ exchange Monte Carlo method was proposed
(Geyer 1991; Hukushima and Nemoto 1996).
The PT algorithm introduces auxiliary distributions with a parameter called the temperature,
generates multiple MCMC samples from target and auxiliary distributions in parallel,
and exchanges the positions of two samples. 
An auxiliary distribution is tempered when the temperature is high
and one with a low temperature is similar to the target distribution.
This ``tempering'' implementation and the exchange process help samples escape from a local mode. 

Auxiliary distributions are also used in other several algorithms.
One example is the Gibbs variable selection in Bayesian variable selection,
where auxiliary distributions approximate the marginal posterior distribution of coefficient parameters for sampling from the joint posterior
(Dellaportas et al.\ 2002).
Another is the cluster Monte Carlo methods, 
efficiently produce samples by block-wise updates 
based on auxiliary distributions 
 (Swendsen and Wang 1987; Higdon 1998). 
These algorithms are referred to as auxiliary variable methods (AVMs) in this paper.

The performance of an AVM depends on both the proposal distribution and the auxiliary distributions.
Hence the parameters of the distributions have to be chosen so that the Markov chains of the AVM mix as faster as possible.
They have been tuned by rough methods or trial-and-error 
in pilot runs so far because their relationship to the mixing speed has not been clear. 

For standard MCMC methods such as Metropolis-Hastings algorithm (Hastings 1970),
Gilks et al.\ (1998) and Haario et al.\ (2001) proposed adaptive MCMC algorithms
that tuned the parameters of a proposal distribution by using past samples during runs.
Haario et al.\ (2001) also proved the convergence theorem of their algorithms, 
which was developed later (Andrieu and Moulines 2006; Roberts and Rosenthal 2007).

In this paper, 
we proposed adaptive AVMs by extending the above adaptive algorithms to general AVMs,
where the algorithms adapt the parameters of the proposal and the auxiliary distributions of AVMs during runs. 
We also proved the convergence theorem of our algorithm in a similar way to Roberts and Rosenthal (2007).

We showed the effectiveness of adaptive AVMs by applying it to the PT algorithm and validating it numerically. That is, 
we developed an adaptive PT algorithm that tune the proposal parameters, which contain the number of temperatures, and temperatures while it runs, and  
  showed the effectiveness of the algorithm via numerical experiments. 
We also proved the convergence of the adaptive PT algorithm by using the theorem of the general adaptive AVMs. 

The rest of the paper is organized as follows.	
the conventional PT algorithm is briefly shown in Section 2
and an adaptive PT algorithm is proposed in Section 3.
Experimental properties of our algorithm are examined in Section 4.
In section 5, the proposed algorithm is extended to general AVMs
and the convergence theorems of our algorithms are proved in Section 6. 
Finally, we give a conclusion in Section 7. 

\section{Parallel Tempering Algorithm}

The PT algorithm is a typical algorithm that uses auxiliary distributions, $\pi_{t_l}(dx_l)$, $l=2,...,L$, where $1 = t_1 > t_2 > \cdots > t_L > 0$. 
The density of the $l$th auxiliary distribution is parametrized by the inverse temperature $t_l$ as $\pi_{t_l}(x) \propto \pi(x)^{t_l}$ or $\pi_{t_l}(x) \propto \pi(x)^{t_l} p(x)^{1-t_l}$, 
where $\pi(x)$ is the density of the target distribution and $p(x)$ is the density of a simple distribution that mix fast by using a conventional MCMC method.
In other words, the inverse temperature $t_l$ tempers the multi-modality of the target distribution $\pi(dx)$
so that the auxiliary densities, $\pi_{t_l}(x_l)$, gradually connect the target density $\pi(x)$ to a simple density $p(x)$ or the uniform distribution.

The PT algorithm executes either of the parallel step and the exchange step at time $n$, with probability $\alpha_r$ and $1-\alpha_r$, respectively.
The parallel step generates the $L$ samples, $x_l^{(n)}$, $l=1, \ldots, L$, according to $\pi_{t_l}(d x_l)$ for each by using a standard MCMC method. 
Note that we employed the Metropolis algorithm with the proposal variance $\gamma_l$ in this paper. 
The exchange step randomly chooses a sample $x_l^{(n)}$ from 
the $L-1$ samples, $x_l^{(n)}$, $l=1, \ldots, L-1$, and exchange $x_l^{(n)}$ for $x_{l+1}^{(n)}$ 
with probability 
\begin{align}
  \label{eq:condition}
  &\min \left(  1,\dfrac {\pi_{t_l} (x_{l+1}^{(n)})\pi_{t_{l+1}}(x_{l}^{(n)})}{\pi_{t_l} (x_l^{(n)})\pi_{t_{l+1}} (x_{l+1}^{(n)})} \right).
\end{align}

The performance of the PT algorithm strongly depends on the inverse temperatures, 
more specifically, their intervals and their number. 
The interval of two adjacent inverse temperatures determines both the similarity of the two distributions and the acceptance probability of an exchange
as seen in Eq.~(\ref{eq:condition}).
The acceptance ratio for the exchanges, which is referred to as exchange ratio in this paper, should not take extreme value. 
For example, Liu (2001) said a preferable value is a half at any interval.
To avoid extreme values and lead to homogeneous exchange ratios, 
Hukushima (1999) updated temperatures using a recursive formula through preliminary runs
and Goswami and Liu (2007) tuned the intervals by iteratively estimating the expected exchange probability through preliminary runs. 

Jasra (2007) treated the intervals as a sequence 
and experimentally compared three inverse-temperature sequences, equal space, logarithmic decay and power decay.  
The results showed the last was best.
Nagata and Watanabe (2008) proved in the low temperature limit when the sequence of inverse temperatures is a geometric progression 
the exchange ratios are homogeneous. 
However, the above methods only discussed the intervals and
did not take into account the proposal distributions,
on which
the mixing of samples and 
 the estimation of exchange ratio also depend. 
In our setting, for example, the Metropolis algorithm has a parameter to be determined,
that is, the proposal variance $\gamma_l$. 
It is necessary to re-set the proposal variance when the inverse temperatures are changed a lot, because the appropriate proposal variances obviously depend on the shape of auxiliary distributions. 

The more auxiliary distributions the PT algorithm has, 
the faster the samples are mixed because flatter auxiliary distributions are available
but the more computational complexity is required.
To solve the trade-off and determine an appropriate number of distributions,
Goswami and Liu (2007) proposed to select the maximum temperature using statistical tests. 
The tests should be done in an off-line manner, that is, they need preliminary experiments in advance.

\section{Adaptive PT Algorithm}

We propose an adaptive PT algorithm that adapts the inverse temperatures, the parameters of proposal distribution,
and the minimum inverse temperature while the algorithm is running.
The three adaptation algorithms are described below.

The exchange ratio should take a moderate value. 
To converge the exchange ratio for $x_{l-1}$ and $x_l$ 
to a specific value, $\alpha\in(0,1)$, typically a half, 
the log inverse temperature, $\zeta_l = \log(t_l)$, is updated as
\begin{align}
  \label{eq:itemplrn}
  \zeta_{l}^{(n+1)} &\leftarrow \zeta_{l}^{(n)} - a^l_n(ER_{l-1,l}^{(n)} - \alpha),
\end{align}
where $ER_{l-1,l}^{(n)}$ is a variable that takes one if the exchange occurs between the samples, $x_{l-1}^{(n)}$ and $x_{l}^{(n)}$, at time $n$, 
and zero otherwise.
The learning coefficient, $a^l_n$, is a decreasing random variable with $n$ that satisfies $\lim_{n \rightarrow \infty} a^l_n = 0 \ almost \ sure$. 

The proposal distribution of the Metropolis algorithm for a target and auxiliary distribution should have an appropriate variance, 
which is an average of the variances of all modes. 
To control the proposal distribution of the Metropolis algorithm for the 
distribution $\pi_{t_l}(d x_l)$ on $\mathbb{R}^p$, 
the the variance $\gamma_l = (\gamma_{l1},...,\gamma_{lp}) \in \mathbb{R}^p$  
and the auxiliary mean parameter $\mu_{l} = (\mu_{l1},...,\mu_{lp}) \in \mathbb{R}^p$ 
are updated as
\begin{equation}
\begin{split}
\mu_{lj}^{(n+1)} & \leftarrow \mu_{lj}^{(n)} + b_n (x_{lj}^{(n+1)} - \mu_{lj}^{(n)}), \\
\gamma_{lj}^{(n+1)} & \leftarrow \gamma_{lj}^{(n)} + 
 b_n \left( (x_{lj}^{(n+1)} - \mu_{lj}^{(n+1)})^2 - \gamma_{lj}^{(n)} \right), \label{eq:learnProp}
\end{split}
\end{equation}
where $x_{lj}^{(n+1)}$ is the $j$th element of $x_{l}^{(n+1)}\in \mathbb{R}^p$.
When $x_l^{(n)}$ is updated to $x_l^{(n+1)}$ by exchanging to $x_{l-1}^{(n)}$ or $x_{l+1}^{(n)}$,	
 the $\mu_l^{(n)}$ jumps to the exchanged value, that is, 
$\mu_l^{(n+1)} \leftarrow x_{l}^{(n+1)}$. 
The learning coefficient, $b_n$, is a decreasing function of $n$ that satisfies $\lim_{n\to\infty} b_n =0$.

The auxiliary distribution with the minimum inverse temperature should be so flat that Metropolis samples can frequently move from one mode to another 
while the total number of auxiliary distributions should be as small as possible.
To determine an appropriate value for the minimum inverse temperature,
the auxiliary distributions $\pi_{t_l}(dx_l)$ with $l>l^*$ are removed 
where $l^*$ is the smallest number that satisfies 
\begin{align}
  \label{eq:criterion}
  \prod_{j=1}^p \gamma_{lj}^{(n)} \geq \prod_{j=1}^p V^{(n)}(x_{lj}) ,
\end{align}
where $V^{(n)}(x_{lj})$ is the sample variance of $x_{lj}$ at time $n$. 
This check is done at time $n = m, 2m, \ldots$, where $m$ is a large number (e.g.~$10^4$).
To improve the reliability, when inequality (\ref{eq:criterion}) holds a few times $d$ (e.g.~$3$) in succession the auxiliary distribution is determined to be enough flat. 

Inequality (\ref{eq:criterion}) shows the relationship between the sample variance and the proposal variance.
Due to Eq.~(\ref{eq:learnProp}), the latter converges to the variance of local region
and hence it is smaller than the sample variance if Metropolis samples are localized in a mode. 
Otherwise, the auxiliary distribution is flat enough.

A pseudo code of the adaptive PT algorithm is given in the following. 
\begin{algorithm}                      
\caption{Adaptive PT algorithm}         
\label{adaptivePT}                       
\begin{algorithmic}
\STATE {\bf Initialize} $x_{l}^{(0)}$, $\zeta_l^{(0)}$, $\gamma_l^{(0)}$, $\mu_l^{(0)}$  and $c_l = 0$ ,$l=1,...,L$. 
($\zeta_1^{(0)}$=0, constant). 
\FOR{$n=0$ to $N-1$}
\STATE $u \sim U[0,1]$ (where $U[0,1]$ is a uniform distribution of the interval (0,1)).
\IF{$u \leq \alpha_r$} 
\FOR{$l=1$ to $L$} 
\STATE ({\it parallel sampling step}) \\ 
\STATE Generate $x_l^{(n+1)}$ via Metropolis Algorithm for $\pi_{t_l^{(n)}}(dx_l)$, which has the proposal variances $\gamma_l^{(n)}$.
\STATE ({\it proposal parameter learning step}) \\
	Update ($\gamma_l^{(n)}$, $\mu_l^{(n)}$) to ($\gamma_l^{(n+1)}$, $\mu_l^{(n+1)}$) by the Eq. (\ref{eq:learnProp}).
\ENDFOR    
\ELSE 
\STATE ({\it exchange step}) \\
     Randomly Choose a neighboring pair, $x_l^{(n)}$ and $x_{l+1}^{(n)}$, and exchange them 
 with the probability Eq.~(\ref{eq:condition}). 
\STATE ({\it inverse temperature learning step}) \\
\STATE Update $\zeta^{(n)}_{l+1}$ to $\zeta^{(n+1)}_{l+1}$ by Eq.~(\ref{eq:itemplrn}). 
\IF{the exchange is accepted, }
\STATE $\mu_k^{(n+1)}$ $\leftarrow$ $x_k^{(n+1)}$, for $k=l,l+1$.
\ENDIF 
\ENDIF 
\STATE ( {\it minimum inverse temperature decision step}) \\
\IF{$(n \bmod m) == 0 $}
\FOR{$l=1$ to $L$}
\STATE If Eq.~(\ref{eq:criterion}) hold, then $c_l \leftarrow c_l + 1$.
\ENDFOR
\STATE $L \leftarrow \min \{ l | c_l \geq d ,l=1,\ldots,L \}$, if $\{ \} \neq \emptyset$. 
\ENDIF
\ENDFOR
\end{algorithmic}
\end{algorithm}

The adaptive PT algorithm converges. 
The proof will be given later as a special case of 
adaptive MCMC algorithms for general auxiliary variable methods. 

\section{Experiments}\label{sec:exmpls}

To confirm the effectiveness of our algorithm,
the following two computer experiments were carried out:
\begin{enumerate}
\item A mixture of four normal distributions.
\item The posterior of a mixture of six normal distributions.
\end{enumerate}

In each of the experiments, 
the burn-in period was a half of the total number of iterations
and sample sets, which are used in an estimation and a scatter plot, are chosen from every 50 samples in post burn-in. 
The proposal distribution of the Metropolis algorithm was an independent normal distribution. 
Other parameters were $\alpha = 0.5$, $\alpha_r = 0.5$, $a^l_n =1/(1 + n/(20 + 10 l)) * \log(\exp(-\zeta_l^{(n)}) + 1)$,
$b_n = 1/(5+0.1n)$, $m=10^4$ and $d=3$.    
$L=25$ and the intervals of inverse temperatures are equal, that is $t^{(0)} = (1,24/25,23/25,\ldots,1/25)$, 
at the initial condition. 
Note that these values are invariant for the each above distribution, i.e., a tuning of these values was not necessary in these experiments.

\subsection{A mixture of four normal distributions}

To see and visualize the properties of our adaptive PT algorithm,
we chose a mixture of four normal distributions in two dimensional space as the target distribution
(Fig.~\ref{fig:coSca2dim}(a)):
\begin{align}
 & g(x) = \notag \\ & \sum_{i=1}^4 \dfrac{1}{4} \dfrac {1}{2\pi^2 {\mathrm det} (\Sigma_i)^{1/2}} 
  \exp \left( - \frac 1 2 (x - \mu_i)^T \Sigma_i^{-1} (x - \mu_i)  \right), \notag
\end{align}
where the parameters of the normal distributions are
\begin{align}
& \mu_1 = (0, 44), & \mu_2 = (44, 0),  \nonumber \\ & \mu_3 = (0, -44), & \mu_4 = (0, -44), \nonumber \\ 
& \Sigma_1  = {\rm diag}(1, 7^2), & \Sigma_2 = {\rm diag}(7^2, 1), \nonumber \\
& \Sigma_3 = {\rm diag}(1, 7^2),  & \Sigma_4 = {\rm diag}(7^2, 1). \nonumber
\end{align}
These normal distributions have quite different variances ($1$ and $7^2$), where the proposal variance learning is difficult.

\begin{figure*}[htbp]
\begin{center}
 \includegraphics[width=0.75\hsize]{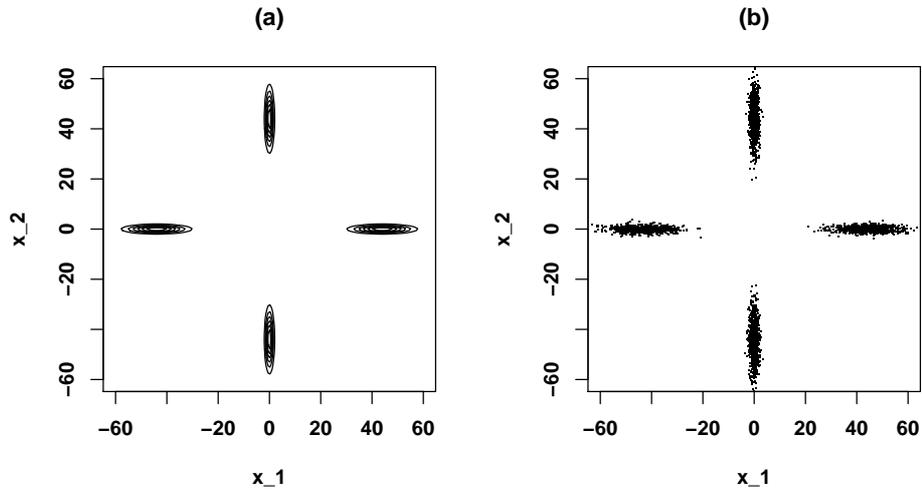}
\caption{A mixture of four normal distributions.
  a) The target distribution.
  b) Samples by the Adaptive PT algorithm.}
\label{fig:coSca2dim}
\end{center}
\end{figure*}

The adaptive PT algorithm ran for $3 \times 10^5$ iterations,
where the auxiliary distributions are $\pi_{t_l}(x) \propto g(x)^{t_l}$ 
and the initial proposal variances are $\gamma_{lj}^{(0)} = 3 \times 10^2$.

As a result, 
our algorithm mixed well and obtained samples from all possible modes (Fig.~\ref{fig:coSca2dim}(b)).
In fact, the number of inverse temperatures was reduced to five after $3 \times 10^4$ iterations 
but the auxiliary distribution $\pi_{\hat t_5}(dx)$ is flat enough (Fig.~\ref{fig:MHminitpr}), where $\hat t_5$ is the $t_5$ obtained by the adaptive PT algorithm. 

The larger the variances of the auxiliary distribution became, the larger the proposal variances became. 
In fact, the sums of proposal variances are $(\hat{\gamma}_{1,1} + \hat{\gamma}_{1,2}, \ldots, \hat{\gamma}_{5,1} + \hat{\gamma}_{5,2}) = (32.26, 41.86, 245.8, 1124, 8704)$.
\begin{figure}[htbp] 
\begin{center}
 \includegraphics[width=0.75\hsize]{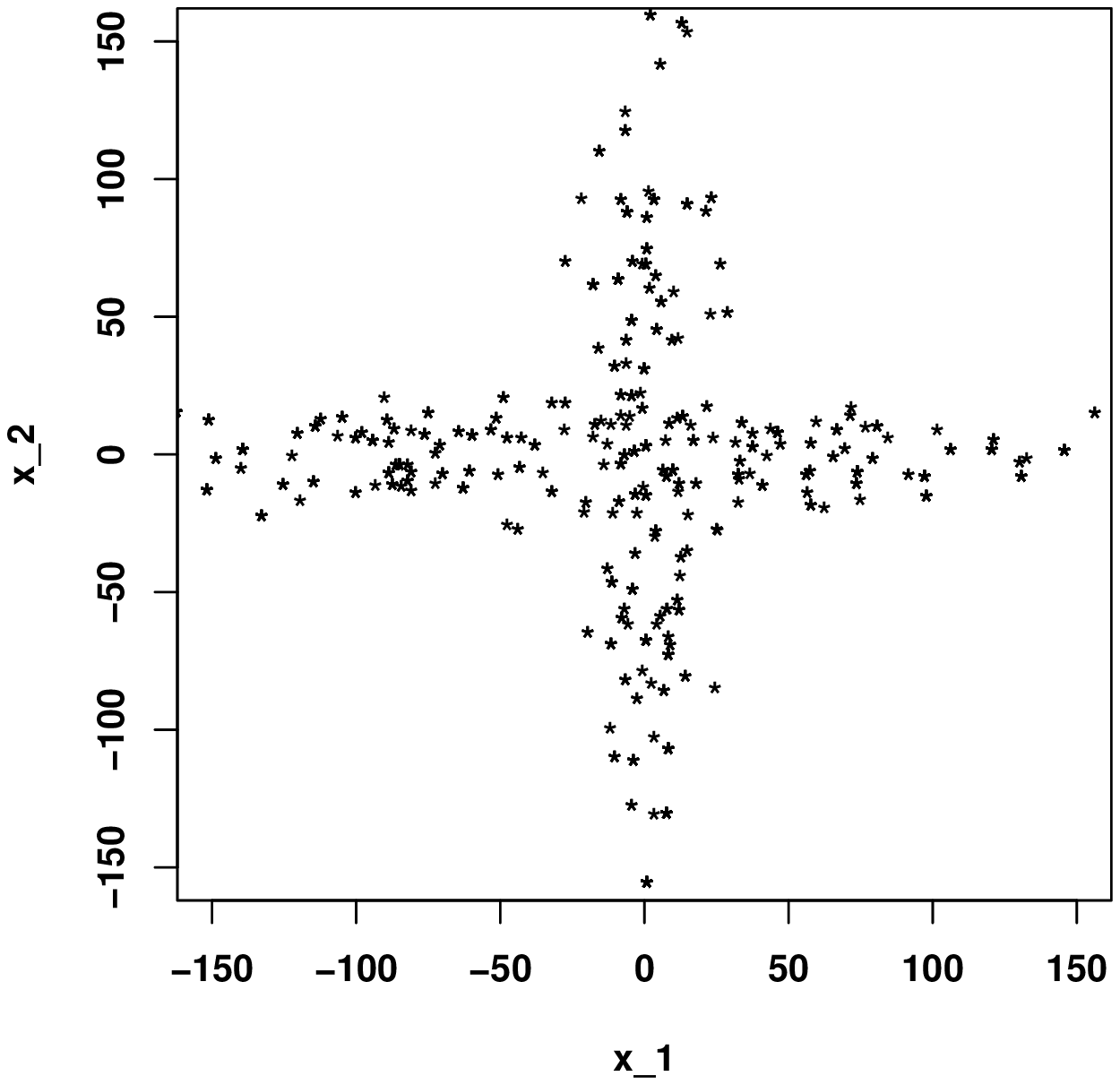}
\end{center}
\caption{Samples by Metropolis algorithm with $\hat {\gamma}_5$ from the auxiliary distribution $\pi_{\hat t_5}(dx)$. They cover all the modes in Fig.~\ref{fig:coSca2dim}.}
\label{fig:MHminitpr}
\end{figure}

The estimated exchange ratios converged to $(0.501 ,0.507 ,0.499, 0.498)$, all of which are almost $\alpha=0.5$. 
Then, the inverse temperatures were $(\hat t_2, \ldots ,\hat t_5) = (0.328,  0.108, 0.0307, 0.00937)$. 

$t_2^{(n)}$ and $\gamma_2^{(n)}$ converge quickly from even the extreme starting points (Fig.~\ref{fig:prIvTPlt}). 
The others also converge as fast as them. 

\begin{figure*}[htbp]
\begin{center}
\includegraphics[width=0.75\hsize]{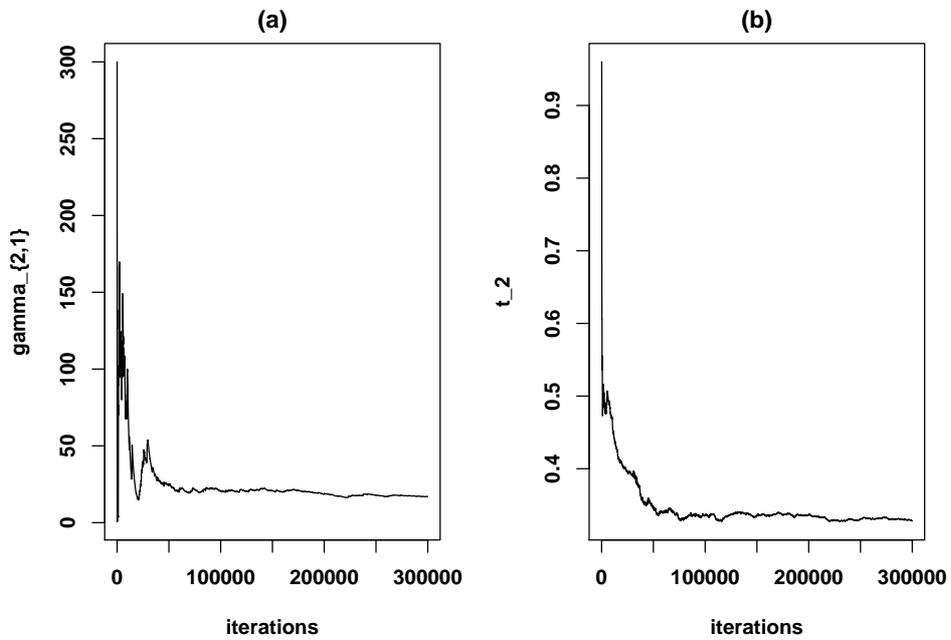}
\caption{trace plot (a):proposal variance $\gamma_{2,1}^{(n)}$, (b):inverse temperature $t_2^{(n)}$.} 
\label{fig:prIvTPlt}
\end{center}
\end{figure*}

\subsection{The posterior of a mixture of six normal distributions}

In the second experiment, 
we estimated the average of component specific means of given the data using Bayesian estimation as is seen in Jasra et al.~(2007). 
The statistical model was a mixture of normal distributions, 
that is, 
\begin{align}
 & f(y|\mu,w,\sigma^2 ) \nonumber \\ &= \sum_{m=1}^M \frac {w_m} {\sqrt{2 \pi} \sigma_m} 
 \exp \left( - \frac 1 {2 \sigma_m^2} \left(y - \mu_m\right)^2 \right), \label{eq:mixDens}  
\end{align}
where $w_M = 1- \sum_{m=1}^{M-1} w_m$.
The priors are a normal-inverse Gamma-Dirichlet prior as follows. 
\begin{align}
& \quad \mu_m \sim N(\xi ,\kappa^2 ), \quad m=1, \ldots, M, \nonumber \\
& \quad \sigma_m^2 \sim IG(\alpha_g,\beta_g), \quad m=1, \ldots, M, \nonumber \\
& \quad w_m \sim {\cal D}(\varrho), \quad m=1, \ldots, M-1, \nonumber
\end{align}
where ${\cal D}(\varrho)$ is the symmetric Dirichlet distribution with parameter $\varrho$.
In the following, the hyper-parameters were $\alpha_g = 12$, $\beta_g = 10$ and $\varrho = 1$.
The parameters $\xi$ and $\kappa^2$ were determined by the median and four times the variance of the given data.

The data of size $150$, $y_{1:150}$, were independently identically distributed according to  
a mixture model of the form (\ref{eq:mixDens}) with parameters, 
$M=6$, $w_1=\cdots=w_6 = 1/6$, $(\mu_1,\ldots,\mu_6) =$ $(-8,-3,1,4,8,13)$, 
$\sigma_1^2 =\sigma_6^2 = 1.5^2$ and $\sigma_2^2=\cdots=\sigma_5^2 = 0.5^2$.
In this case, the posterior $\pi(\mu,w,\sigma^2|y_{1:150})$ was a 17 dimensional distribution
and had $6!=720$ symmetric modes due to the invariance against permutation of the labels of the parameters. 

The auxiliary distributions were set to 
\begin{align}
  \pi_{t_l}(\mu,w,\sigma^2|y_{1:150}) \propto \left (\prod_{i=1}^{150} f(y_i|\mu,w,\sigma^2) \right )^{t_l} p(\mu,w,\sigma^2),  \notag
\end{align}
where $p(\mu, w, \sigma^2)$ was the prior.

Our algorithm was compared to the conventional PT algorithm with the fixed parameters.
The parameter values of the conventional algorithm were shifted from the values obtained by the adaptive PT algorithm as follows.
\begin{enumerate}
\item[(a)] $\zeta_l \leftarrow \hat {\zeta}_l \cdot \varphi_{\zeta}$, for $l=1,\ldots,L$, $(0.5 \leq \varphi_{\zeta} \leq 3)$. \\ $L \leftarrow \hat L$, $\gamma_l  \leftarrow \hat {\gamma}_l$.
\item[(b)] $\gamma_l \leftarrow \hat {\gamma}_l \cdot \varphi_{\gamma}^2$, for $l=1,\ldots,L$, $(0.1 \leq \varphi_{\gamma} \leq 3)$. \\ $L \leftarrow \hat L$, $\zeta_l  \leftarrow \hat {\zeta}_l$.
\item[(c)] $L \leftarrow \hat L + \varphi_{L}$, $(-5 \leq \varphi_{L} \leq 5)$. \\ $\gamma_l \leftarrow \hat \gamma_l$, $\zeta_l  \leftarrow \hat {\zeta}_l$.
\\ (If $\varphi_L > 0$, $\zeta_l$ and $\gamma_l$ were learned, for $l=\hat L + 1,..., \hat L + \varphi_L$, to maintain the fairness.)
\end{enumerate}

We ran the adaptive PT algorithm and the conventional PT algorithms for $10^6$ iterations respectively. 
The initial condition were $w_{l,m}^{(0)}=1/6$, 	
$\sigma_{l,m}^{2(0)} \sim IG(\alpha_g,\beta_g)$, $\mu_{l,m}^{(0)} \sim U[\min(y_{1:150}),\max(y_{1:150})]$ for each run. 
The initial values of $\gamma_{1}^{(0)}, \ldots, \gamma_{L}^{(0)}$ were the sorted $L$ random numbers from $U[0.0001,800]$.
The variables of posterior were divided into four blocks, the numbers of which are (5,4,4,4). 
Each Metropolis algorithm updated for the every block.

We evaluated the estimators of $\mu_m$, $m=1,\ldots,6$. 
The accuracy of the estimation was evaluated by the root mean square error (RMSE).
To evaluate the total error of the six estimators, RMSE takes the root average of the errors of the six estimators, that is, 
\begin{align}
  \mathrm{RMSE}(i) &= \left( \frac{1}{6} \sum_{m=1}^{6} (\bar \mu_m(i) - 2.5)^2 \right)^{1/2}, \notag
\end{align}
where $\bar \mu_m(i)$ is the estimator, which is the average of samples $\mu_m^{(n)}$, in the $i$th trial, and 2.5 is the true value. 

As a result, 
our algorithm can obtain appropriate parameters and achieve very low RMSEs.
Fig.~\ref{fig:wholeRMSE}(a) and (b) show 
the RMSEs of the adaptive PT algorithm and the conventional PT algorithms in 50 runs
were less than those of the conventional PT algorithms with shifted parameters.
On the other hand, even if the number of temperatures increases the RMSEs don't increase (Fig.~\ref{fig:wholeRMSE}(c)), but 
the computational costs of the algorithms increase.
In fact, the shifted inverse temperatures could not control the exchange ratios well
(Fig.~\ref{fig:avexrate}).
\begin{figure*}[htbp]

 \begin{tabular}{cc} 
    \begin{minipage}[b]{0.5\linewidth} 
	    \subfigure[]{%
	        \includegraphics[clip, width=\linewidth]{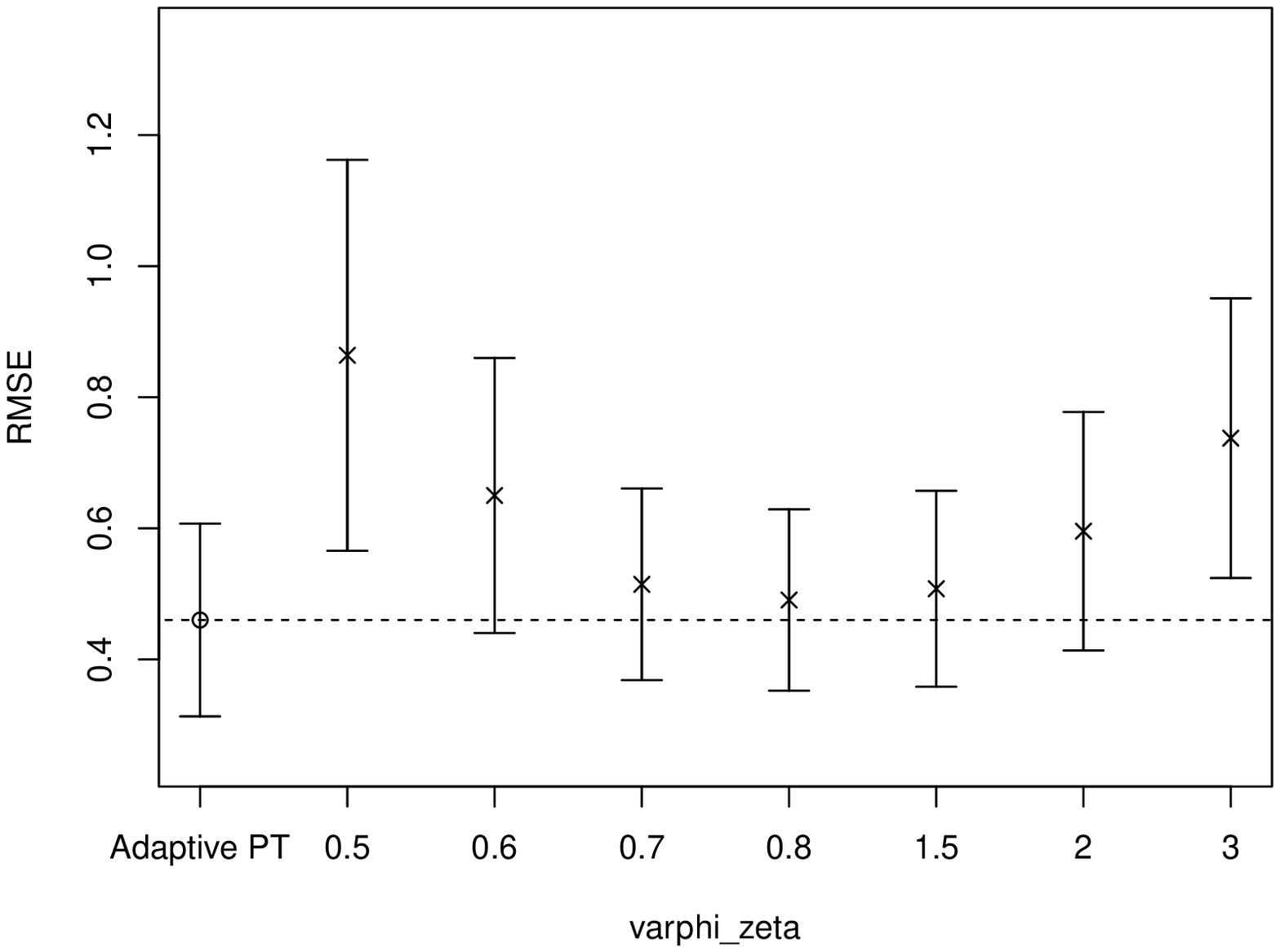}%
	     \label{fig:RmseTmpr}}
  \end{minipage}

    \begin{minipage}[b]{0.5\linewidth} 
	    \subfigure[]{%
	        \includegraphics[clip, width=\linewidth]{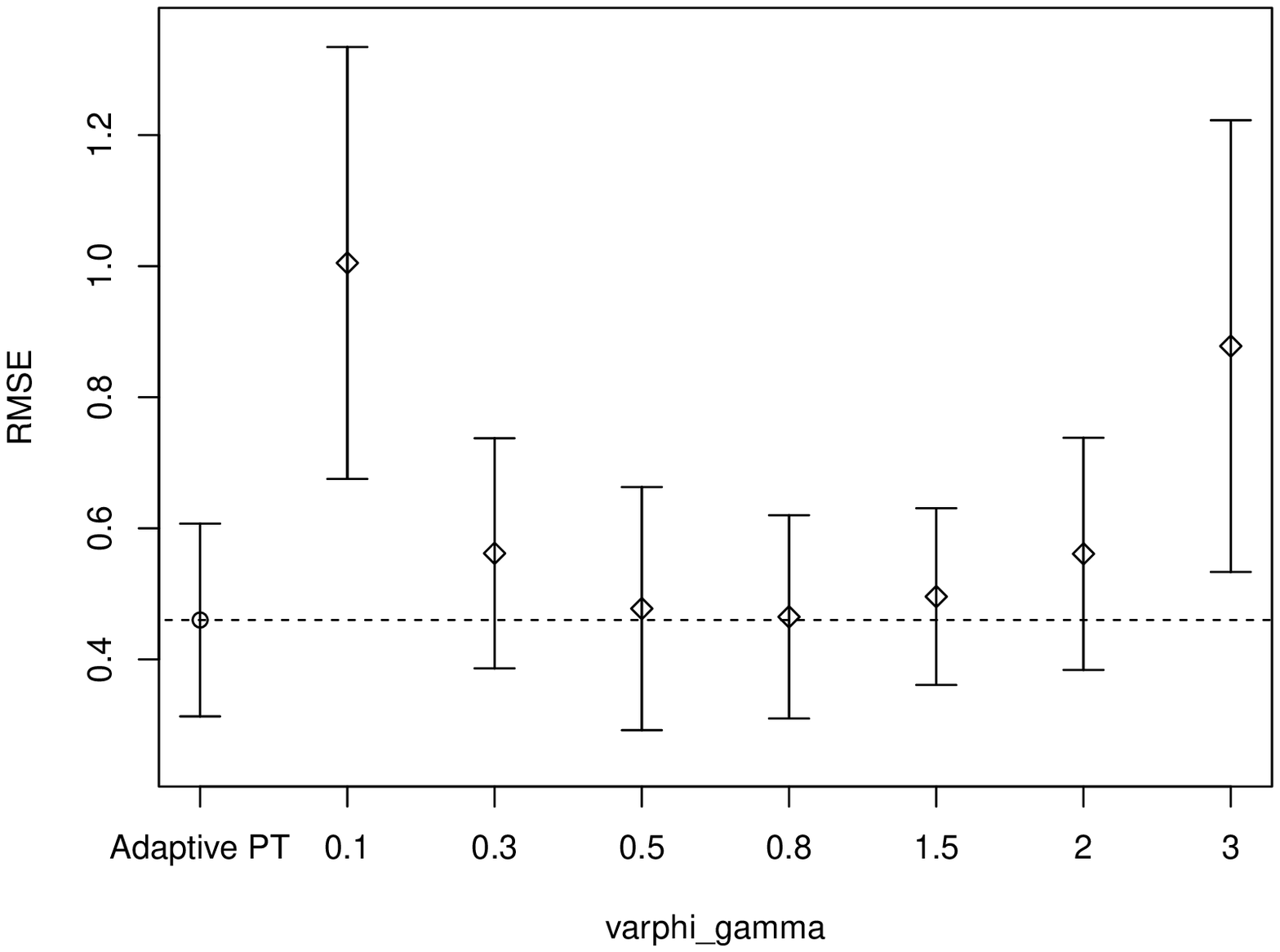}%
	     \label{fig:RmsePrp}}
  \end{minipage}
  
\end{tabular}
    \begin{minipage}[b]{\linewidth} 
    \begin{center}
	    \subfigure[]{%
	        \includegraphics[clip, width=0.5\linewidth]{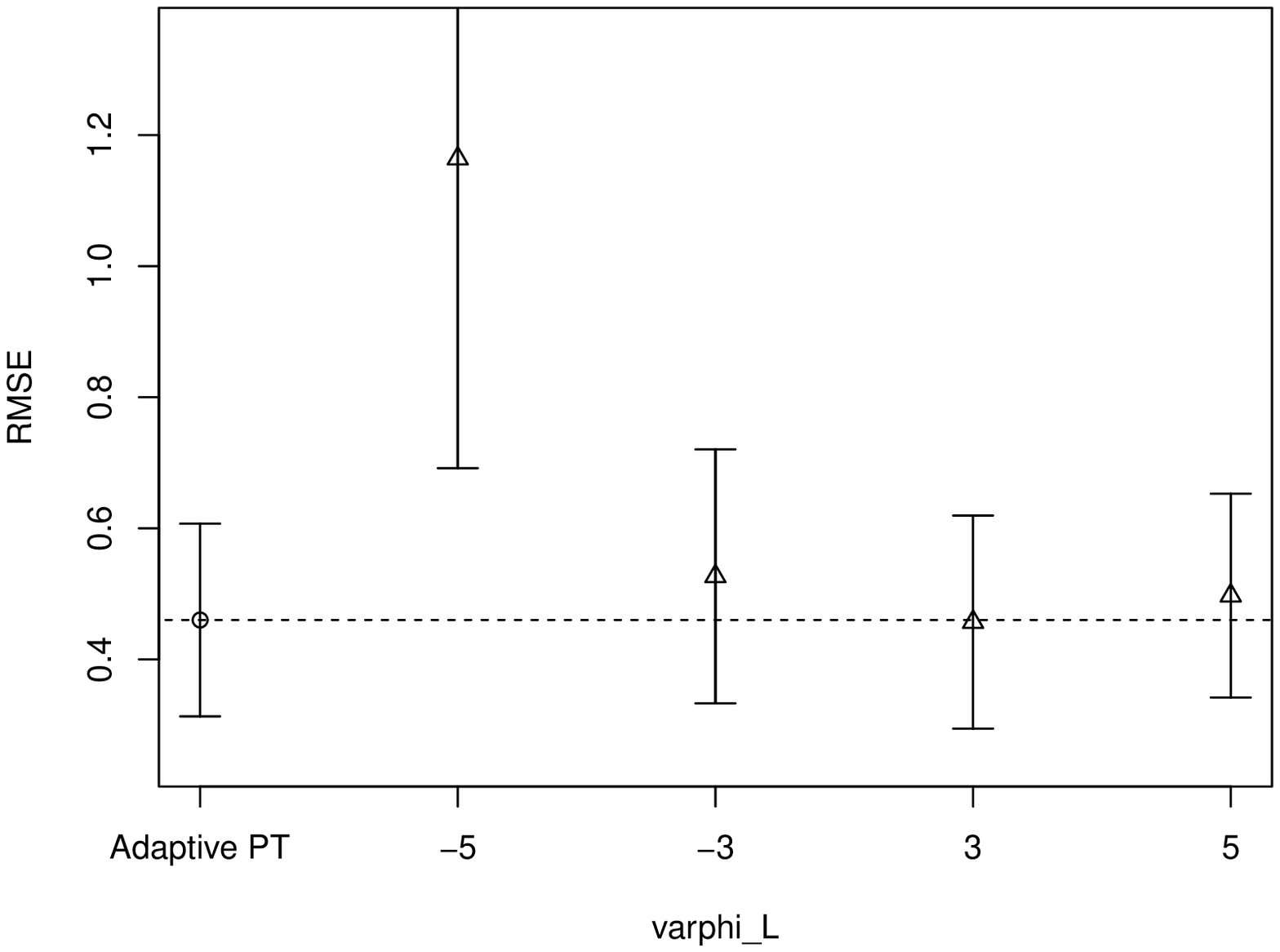}%
	     \label{fig:RmseL}}
	\end{center}
  \end{minipage}
\caption{
RMSEs in 50 runs. 
Each plot displays an average and a standard error of RMSEs for each algorithm in a mark and a radius of the error bar, respectively. 
The adaptive PT algorithm : ($\circ$). The conventional PT algorithms with the shifted parameters are separately plotted in these figures.
The inverse temperatures : $\times$ in (a), The proposal variances : $\diamond$ in (b), The number of inverse temperatures : $\triangle$ in (c). 
}
\label{fig:wholeRMSE}

\end{figure*}

\begin{figure}[htbp]
\begin{center}
 \includegraphics[width=\hsize]{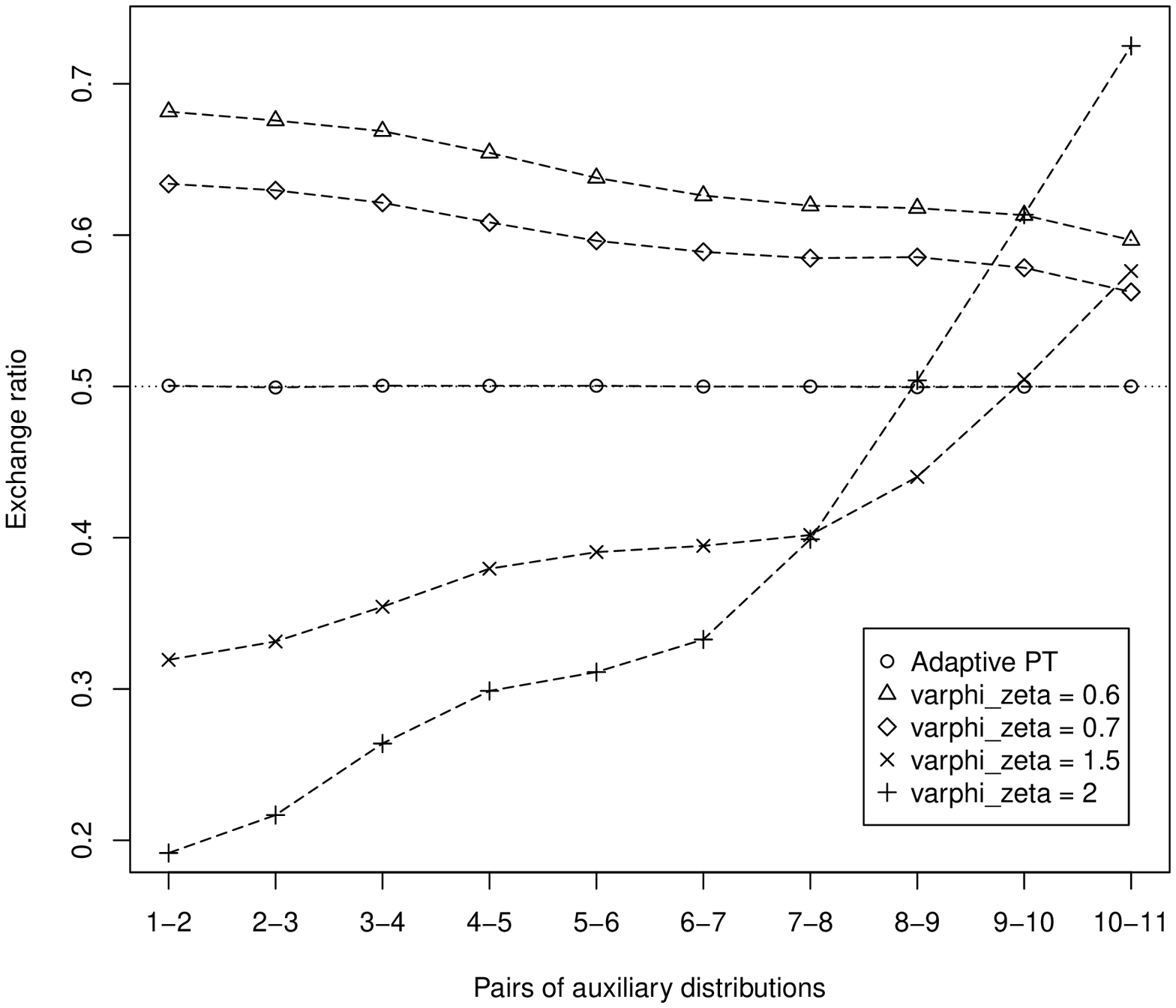}
\caption{The averages of the estimated exchange ratios in 50 runs. 
The ratios of the adaptive PT algorithm were almost the predefined value
while the conventional PT algorithm with shifted inverse temperatures could not be tuned constant and a moderate value (0.5), a dotted line.}
\label{fig:avexrate}
\end{center}
\end{figure}

\section{Generalization to Auxiliary Variable Methods}\label{sec:AVMoutline}

The idea of the adaptive PT algorithm is applicable to general AVMs.
AVMs are mathematically formulated as below.

Let $\pi(dx)$ be a distribution on a state space ${\cal X}$ with $\sigma$-algebra $F_{\cal X}$
and $\pi_{\lambda}(dy|x)$ be a conditional distribution on a state space ${\cal Y}$ with $\sigma$-algebra $F_{\cal Y}$ given $F_{\cal X}$,
where $\lambda \in \Lambda$ is a parameter vector.
Then, the marginal distribution on ${\cal X}$ of the joint distribution $\pi_{\lambda}(dx,dy) = \pi_{\lambda}(dy|x)\pi(dx)$ is $\pi(dx)$
irrespective of $\pi_{\lambda}(dy|x)$.

In case of MCMC methods with auxiliary variables, 
$\pi(dx)$ corresponds to the target distribution and $\pi_{\lambda}(dy|x)$ to the auxiliary distributions.
We term an MCMC method that draw samples $(x',y')$ from $\pi_{\lambda}(dx,dy)$ to obtain $x'$ an auxiliary variable method. 
In the PT algorithms, for example, the auxiliary distributions are $\pi_{\lambda}(dy|x)=\prod_{l=2}^L \pi_{t_l}(d x_l)$, $\lambda=(t_2, \ldots, t_L)$ and the auxiliary variables 
 $y = (x_2,...,x_L)$.

In order to introduce adaptation, we need to consider time-varying parameters. 
Let $\{ P_{\theta}((x,y),(dx,dy)) \}_{\theta \in \Theta}$ be a family of Markov transition kernels 
 on ${\cal X} \times {\cal Y}$
with stationary distribution $\pi_{ \lambda}(dx,dy)$, that is,
\begin{align}
  \left(\pi_{ \lambda} P_{ \theta}\right)(dx, dy)
  &= \iint_{x',y'} \pi_{ \lambda}(dx',dy') P_{ \theta}\left( (x',y'), (dx,dy)\right) \nonumber \\ 
  &= \pi_{ \lambda}(dx,dy), \nonumber
\end{align}
where $\lambda \subseteq \theta$. 
Then, the adaptive MCMC for AVM updates the parameters $\theta$
during generating chains $(x^{(n)},y^{(n)})$ by $P_{\theta}$  
as the following pseudo code. 
\begin{algorithm}                      
\caption{Adaptive MCMC for AVM}         
\label{generalAVM}                       
\begin{algorithmic}
\STATE {\bf Initialize} $(x^{(0)},y^{(0)}),\theta^{(0)}$.
\FOR{$n=0$ to $N-1$}
\STATE [1] $(x^{(n+1)},y^{(n+1)}) \sim  P_{\theta^{(n)}}((x^{(n)},y^{(n)}),(dx,dy))$ 
\STATE [2] Update $\theta^{(n)}$ to $\theta^{(n+1)}$ by using the result of step 1. such as $(x^{(n+1)},y^{(n+1)})$.
\ENDFOR
\end{algorithmic}
\end{algorithm}

In the adaptive PT algorithms, for example,
the time-varying parameter vector is $\theta = (\gamma_1,\ldots,\gamma_L ; t_2, \ldots, t_L)$. 

\section{Convergence Theorem}\label{sec:ConThm}

Atchade (2011) and Fort et al.~(2011) proved convergence theorems of adaptive MCMC algorithms
that learn the parameters of the target distribution. 
The conditions for convergence in their theorems are, however, technical and strict.
For example, the stationary distribution must converge.
These conditions will considerably restrict the available learning algorithms.

In this section, we show some convergence theorems
that our algorithm in the previous section converges under weaker conditions.
Here, convergence means that an algorithm is ergodic,
that is,
\begin{align}
 \lim_{n \rightarrow \infty} \| A^{(n)}(( x, y, \theta),dx) - \pi(dx) & \| = 0 , \nonumber \\
  & \forall (x, y) \in {\cal X \times \cal Y}, \theta \in \Theta, \nonumber 
\end{align}
where $\|\mu(dx) - \nu(dx)\| = {\rm sup}_{A \in \cal F_{\cal X} }| \mu(A) - \nu(A) |$
and
\begin{align}
 & A^{(n)}((x, y, \theta),B_{\cal X}) \nonumber \\
 &= P\left[
    x^{(n)} \in B_{\cal X} |x^{(0)} = x,y^{(0)} = y, \theta^{(0)} = \theta \right], \nonumber \\[-5pt] 
& \qquad \qquad \qquad \qquad \qquad \qquad \qquad \qquad \qquad  
B_{\cal X} \in F_{\cal X}. \notag
\end{align}

\begin{theorem}
The adaptive MCMC for AVM is ergodic if the following conditions hold:
\begin{description}
\item[(a) Simultaneous uniform ergodicity]
\begin{align}
& \forall \varepsilon > 0, \ \exists N \in \mathbb{N} \ s.t. \nonumber \\ 
& \quad ||P_{ \theta}^N(( x, y),dx) - \pi(dx)|| \leq \varepsilon,  \nonumber \\
& \qquad \qquad \qquad \qquad \quad \forall (x, y) \in {\cal X \times \cal Y}, \theta \in \Theta.
\end{align}

\item[(b) Diminishing adaptation] 
  \begin{align}
  & {\lim}_{n \rightarrow \infty}  
    \sup_{(x,y) \in {\cal X} \times {\cal Y}} \notag \\
  &  \| P_{\theta^{(n+1)}} \left( (x, y),(dx,dy) \right) - P_{\theta^{(n)}}\left((x, y),(dx,dy)\right) \| \notag \\
    & =0 \ in \ probability.
  \end{align}
\end{description}
\end{theorem}
\begin{proof}
See Appendix \ref{sec:proofTh1}.
\end{proof} 
The above conditions do not require
that the auxiliary parameter $\lambda^{(n)}$ and the stationary distribution $\pi_{\lambda^{(n)}}$ converge.
The condition (a) can be replaced with more concrete condition that checks only properties of the Markov transition kernel as follows.
\begin{description} 
\item[(a')] (Simultaneously strongly aperiodically geometrical ergodicity) 
There exists 
$C \in {\cal F_{{\cal X }\times {\cal Y}}}, \ V:{\cal X} \times {\cal Y} \rightarrow [1,\infty)$
, $\delta > 0$, $ \tau < 1$, and $b < \infty$, such that ${\rm sup}_C V < \infty$ 
and the following conditions hold for all $\theta \in \Theta$.

\begin{description}
\item[(i) (Strongly aperiodic minorisation condition)]
There exist a probability measure $\nu_{\theta}\left (dx,dy \right)$ on $C$ such that  
$P_{ \theta} (( x, y),(dx',dy')) \geq \delta \nu_{\theta}( dx',dy' )$ for all 
$x, y \in C$.
\item[(ii) (Geometric drift condition)]
\begin{align}
\left(P_{ \theta}V \right)( x,y) \leq \tau V(x,y) + & b \bm 1_{\{C\}}( x, y), \nonumber \\
& for \ all \ x,y \in {\cal X \times \cal Y}, \nonumber 
\end{align}
where $\ (P_{ \theta}V )( x, y) $ \\ $\equiv
 \int\!\!\!\int P_{ \theta} \left (( x, y),(dx',dy') \right) V(x',y') dx' dy'$, 
and \\ $\bm 1_{\{ \cdot \}}(x)$ is the indicator function.
\end{description}
\end{description}
\begin{theorem}
The adaptive MCMC for AVM is ergodic if the condition (b) in Theorem 1, the condition (a')
 and $E[V(x^{(0)},y^{(0)})] < \infty$ hold. 
\end{theorem}
\begin{proof} 
Straightforward from Proposition 3 and the proof of Theorem 3 in Roberts and Rosenthal (2007), and Theorem 1.
\end{proof} 

\begin{theorem}[Weak law of large numbers]
Suppose an adaptive MCMC for AVM satisfies the conditions (a) and (b)
and let $g:\cal X \rightarrow \mathbb{R}$ be a bounded measurable function.
Then,
\begin{align}
\frac 1 n \sum_{i=1}^n g(x^{(i)}) \rightarrow \int g(x)\pi(dx) \quad \mbox{ in probability} \notag
\end{align}  
as $n \rightarrow \infty$
for any initial values $(x, y) \in \cal X \times \cal Y$ and $\theta \in \Theta$.
\end{theorem}
\begin{proof}
Straightforward from the coupling argument (Roberts and Rosenthal 2007). 
\end{proof}

The convergence of the adaptive PT algorithm is proved by applying Theorem 2 as below.
\begin{theorem}
  The adaptive PT algorithm is ergodic if the following conditions hold:
  \begin{description}
  \item[(s1)] The support $S$ of the target distribution $\pi(d x)$ is compact and the density $\pi(x)$ is continuous and positive on $S$.
  \item[(s2)] The family of proposal densities $\{ q_{\gamma} \}_{\gamma \in \Gamma^p}$ is continuous and positive on $S^2  \times \Gamma^p$, where $\Gamma = [c,C]$.
\end{description}
\end{theorem} 
\begin{proof}
See Appendix \ref{sec:proofTh4}.
\end{proof}
It will be possible to remove the assumption that $S$ is compact by extending Theorem 6 of Bai et al.~(2011).
\section{Conclusions}

This paper proposed adaptive MCMC algorithms
to learn parameters of proposal distributions and auxiliary distributions simultaneously, 
and proved convergence theorems that give weak sufficient conditions for convergence.

We applied this framework to the Parallel Tempering algorithm
and showed that the adaptive PT algorithm can adapts its parameters on the fly so that samples mix rapidly 
 by experiments with a mixture model. 
We also presented that the performance of the PT algorithm depends on its parameters
and the adaptive PT algorithm finds good parameters through experiments for Bayesian estimation.

Although we discussed the PT algorithm in a real space so far, 
we consider the idea of adaptation is applicable to those in a discrete space.
We also consider our adaptive framework is applicable to other auxiliary variable methods such as the Gibbs variable selection, the partial decoupling method (one of the cluster Monte Carlo methods) and so on.
We will extend our theory to the new field in the future.

\section*{Appendix}

\appendix  

\section{Proof of Theorem 1}\label{sec:proofTh1}

Let $\epsilon > 0$, and choose $N \in \mathbb{N}$ as in condition (a).
From condition (b) and the coupling argument in the proof of Theorem 1 of Roberts and Rosenthal (2007), the following result 
holds.  

There exists $n^* \in \mathbb{N}$ large enough so that for $K > n^* + N$, 
there exists a second chain $\{ x'^{(n)},y'^{(n)} \}_{n=K-N}^K$, such that $(x'^{(K-N)},y'^{(K-N)}) = (x^{(K-N)},y^{(K-N)})$ and \\
$(x'^{(n+1)},y'^{(n+1)}) \sim P_{\theta^{(K-N)}}((x'^{(n)},y'^{(n)}),dx,dy)$ for $n=K-N,...,K-1$, 
and $P(x^{(K)} \neq x'^{(K)} ) \leq 2 \epsilon$. 

Then it follows that 
\begin{equation}
|| P(x^{(K)} \in dx ) - P(x'^{(K)} \in dx )|| \leq 2 \epsilon, \label{eq:tdis_xx'} 
\end{equation}
 where $P(x^{(K)} \in dx)$ denotes the distribution of $x^{(K)}$.
 (Because of $||P(y \in dx) - P(z \in dx)|| \leq P(y \neq z)$.)

On the other hand, from the condition (a), for all $A_{\cal X} \in F_{\cal X}$, we have
\begin{align}
\epsilon \geq & \left | E [P^N_{\theta^{(K-N)}}( (x^{(K-N)},y^{(K-N)}),A_{\cal X}) - \pi(A_{\cal X})] \right | \notag \\
= & \left | P(x'^{(K)} \in A_{\cal X}) - \pi(A_{\cal X}) \right |. \notag
\end{align}
That is,
\begin{align}
|| P(x'^{(K)} \in dx ) - \pi(dx) || \leq \epsilon . \label{eq:tdis_x'pi}
\end{align}

From inequality ($\ref{eq:tdis_xx'}$) and ($\ref{eq:tdis_x'pi}$), we have 
\begin{align}
|| P(x^{(K)} \in dx ) - \pi(dx) || \leq 3 \epsilon. \label{eq:tdis_xpi}
\end{align}
Since $K \geq n^* + N$ is arbitrary, the algorithm is ergodic.

\section{Proof of Theorem 4}\label{sec:proofTh4}

We prove the sufficient conditions of convergence in Theorem 2 are satisfied. Firstly, we prove the condition (a') holds.

Let Borel $\sigma$-algebra on $\mathbb R^p$ be ${\cal B}(\mathbb R^p)$.  
For $\bm x \in S^L $, $\bm \gamma \in \Gamma^{pL}$, $\bm t \in {\cal T}^L$
and $\bm B = B_1 \times B_2 \times \cdots \times B_L, \ B_l \in {\cal B}(S)$,
the transition kernel of the PT algorithm is  
\begin{align}
K_{\bm \gamma,\bm t}(\bm x,\bm B ) = \alpha_r \prod_{l=1}^{L} P_{\gamma_l, t_l}(x_l,B_l) + 
(1-\alpha_r) \sum_{l=2}^L \varsigma_l k_{l,l-1}(\bm x,\bm B), \label{eq:KerPT} 
\end{align}
where $0 \leq \varsigma_l \leq 1$, $\sum_{l=2}^L \varsigma_l = 1$, $P_{\gamma_l,t_l}(x_l,dx_l)$ and 
$k_{l,l-1}(\bm x,d \bm x')$ are the Metropolis transition kernel for $\pi_{t_l}(dx_l)$ and the transition kernel of 
an exchange process of $x_l$ and $x_{l-1}$, respectively.

By condition (s1), we have $d \equiv \sup_{x \in S,t \in {\cal T}} \pi_t(x) < \infty$.
By the compactness of $S$ and condition (s2), we have also $\delta \equiv$ $\inf_{x,x' \in S, \gamma \in \Gamma^p}$ $q_{\gamma}(x,x') > 0$.

For $x \in S$ and $t \in {\cal T}$, denote $R_{x,t} = \left \{ y \in S | \frac{\pi_{t}(y)}{\pi_{t}(x)} \leq 1 \right \}$.
For $x_l \in S$, $B_l \in {\cal B}(S)$,  
$t_l \in {\cal T}$ and $\gamma_l \in \Gamma$, we have 
\begin{align}
& P_{\gamma_l, t_l}(x_l,B_l) \notag \\
& = \int_{B_l} q_{\gamma_l}(x_l,x_l') \min \left (1,\frac{\pi_{t_l}(x'_l)}{\pi_{t_l}(x_l)} \right )  d x_l' \notag \\ 
& + \bm 1_{\{B_l\}}(x_l) \int_S q_{\gamma_l}(x_l,\tilde x_l) \left \{ 1 - \min \left ( 1,\frac{\pi_{t_l}(\tilde x_l)}{\pi_{t_l}(x_l)} \right ) \right \} d \tilde x_l \notag \\
& = \int_{B_l \cap R_{x_l,t_l}}  q_{\gamma_l}(x_l,x_l') \frac{\pi_{t_l}(x'_l)}{\pi_{t_l}(x_l)} d x_l' \notag \\
& + \int_{B_l \cap R_{x_l,t_l}^c} q_{\gamma_l}(x_l,x_l') d x_l' \notag \\
& \geq \frac{\delta}{d} \int_{B_l \cap R_{x_l,t_l}} \pi_{t_l}(x'_l) d x_l' + 
\frac{\delta}{d} \int_{B_l \cap R_{x_l,t_l}^c} \pi_{t_l}(x'_l)  d x_l' \notag \\
& = \frac{\delta}{d} \pi_{t_l}(B_l). \notag
\end{align}
From Eq.~(\ref{eq:KerPT}), this inequality leads to 
\begin{align}
K_{\bm \gamma,\bm t}(\bm x,\bm B ) & \geq \alpha_r \prod_{l=1}^{L} P_{\gamma_l, t_l}(x_l,B_l) \notag \\
& \geq \alpha_r \prod_{l=1}^{L} \frac {\delta}{d} \pi_{t_l}(B_l) \notag \\
& =\alpha_r \frac { \delta^L}{d^L} \pi_{\bm t}(\bm B), \label{eq:KerIneq}
\end{align}
where 
$\pi_{\bm t}(\bm B) = \prod_{l=1}^{L} \pi_{t_l}(B_l)$ is a probability measure on $S^L$.
Since the inequality (\ref{eq:KerIneq}) holds for all $\bm B \in {\cal B}(S^L)$, the condition (a')(i) follows.

Let $0 < \tau < 1$, $V(\bm x) =1$ if $\bm x \in S^L$, otherwise $V(\bm x) \equiv 1/\tau$, and $b=1-\tau$.
Then we have
\begin{align}
( K_{\bm \gamma,\bm t} V )(\bm x) \leq \tau V(\bm x) + b \bm 1_{\{S^L\}}(\bm x), \quad \forall \bm x \in \mathbb{R}^{pL}.
\end{align}
This inequality implies that the condition (a')(ii) is satisfied.
Also we have $E[V(x^{(0)},y^{(0)})] \leq 1/\tau < \infty$.

From Eq.~(\ref{eq:itemplrn}) and (\ref{eq:learnProp}),
 it follows that $t_l^{(n+1)} - t_l^{(n)} \rightarrow 0 \ a.s.$ and $\gamma_l^{(n+1)} - \gamma_l^{(n)} \rightarrow 0$ as $n \to \infty$. 
The minimum inverse temperature decision process changes the value of $\varsigma_l$ only finite times. 
Thus, the condition (b) in Theorem 1 holds. 

The proof is complete.

\nocite{Andrieu_etal06,Atchade11,Bai_etal11,SwendsenWang87,Haario01,Hastings70,Higdon98,Hukushima96,geyer91arti,Goswami07,Fort11AS,Gilks98,Jasra07SC,Liu01,Nagata08,RobertCPtext04,Roberts07,Dellaportas02} 

\bibliographystyle{spmpsci}      

\bibliography{English.bib}   

\end{document}